\documentclass[12pt]{amsart}
\usepackage{amssymb,amscd}
\usepackage{verbatim}

\usepackage{xcolor}

\usepackage{amsmath,amssymb,graphicx,mathrsfs}   
\usepackage{enumerate}
\usepackage[colorlinks=true,allcolors = blue]{hyperref} 

\usepackage{tikz}
\usetikzlibrary{matrix}

\usepackage[all]{xy}

\usepackage{amssymb,amsfonts,amsthm,amsmath,calligra}
\usepackage{slashed}
\usepackage{yfonts}
\usepackage{mathrsfs,pifont}
\usepackage{float}

\usepackage[all]{xy}

\usepackage{tikz}

\usepackage{graphicx}
\usepackage{xcolor}

\usepackage{amssymb,amsfonts,amsthm,amsmath}
\usepackage[all]{xy}
\usepackage{slashed}
\usepackage{yfonts}
\usepackage{mathrsfs,pifont}

\let\frak\mathfrak

\def\>{\relax\ifmmode\mskip.666667\thinmuskip\relax\else\kern.111111em\fi}
\def\<{\relax\ifmmode\mskip-.333333\thinmuskip\relax\else\kern-.0555556em\fi}
\def\vsk#1>{\vskip#1\baselineskip}
\def\vv#1>{\vadjust{\vsk#1>}\ignorespaces}
\def\vvn#1>{\vadjust{\nobreak\vsk#1>\nobreak}\ignorespaces}

  \let\ssize\scriptstyle
\let\sssize\scriptscriptstyle

\let\Medskip\medskip
\def\medskip{\par\Medskip}
\let\Bigskip\bigskip
\def\bigskip{\par\Bigskip}

\let\Maketitle\maketitle
\def\maketitle{\Maketitle\thispagestyle{empty}\let\maketitle\empty}

\newtheorem{thm}{Theorem}[section]

\newtheorem{prop}[thm]{Proposition}

\newtheorem{defn}[thm]{Definition}

\theoremstyle{definition}                                  
\numberwithin{equation}{section}

\theoremstyle{definition}
\newtheorem*{rem}{Remark}
\newtheorem*{example}{Example}

\let\mc\mathcal
\let\nc\newcommand

\let\la\lambda

\let\phi\varphi

\let\om\omega

\let\der\partial

\let\geq\geqslant

\let\leq\leqslant

\let\on\operatorname
\let\bi\bibitem
\let\bs\boldsymbol

\def\C{{\mathbb C}}
\def\Z{{\mathbb Z}}

\def\F{{\mathbb F}}

\def\+#1{^{\{#1\}}}

\def\Gr{\on{Gr}}

\def\beq{\begin{equation}}
\def\eeq{\end{equation}}
\def\be{\begin{equation*}}
\def\ee{\end{equation*}}

\nc{\bea}{\begin{eqnarray*}}
\nc{\eea}{\end{eqnarray*}}
\nc{\bean}{\begin{eqnarray}}
\nc{\eean}{\end{eqnarray}}

\let\ga\gamma

\nc{\Il}{{\mc I_{\bs\la}}}
\nc{\bla}{{\bs\la}}
\nc{\Fla}{\F_\bla}
\nc{\tfl}{{T^*\Fla}}
\nc{\GL}{{GL_n(\C)}}
\nc{\GLC}{{GL_n(\C)\times\C^*}}

\let\sd s 

\def\ddk_#1{\kk_{#1}\<\>\frac\der{\der\<\>\kk_{#1}}}

\def\bul{\mathbin{\raise.2ex\hbox{$\sssize\bullet$}}}
\def\intt{\mathchoice
{\mathop{\raise.2ex\rlap{$\,\,\ssize\backslash$}{\intop}}\nolimits}
{\mathop{\raise.3ex\rlap{$\,\sssize\backslash$}{\intop}}\nolimits}
{\mathop{\raise.1ex\rlap{$\sssize\>\backslash$}{\intop}}\nolimits}
{\mathop{\rlap{$\sssize\<\>\backslash$}{\intop}}\nolimits}}

\let\kk q 
\let\cc c

\let\Ko K

\def\GZ/{Gelfand-Zetlin}
\def\KZ/{{\slshape KZ\/}}
\def\qKZ/{{\slshape qKZ\/}}
\def\XXX/{{\slshape XXX\/}}

\nc{\A}{{\mc A}}

\nc{\hsl}{\widehat{{\frak{sl}_2}}}

\nc{\BC}{{ \mathbb C}}
\nc{\lra}{\longrightarrow}
\nc{\CO}{{\mathcal{O}}}
\nc{\BZ}{{ \mathbb Z}}
\nc{\hfn}{\hat{\frak{n}}}
\nc\Zs{{\Z/p^s\Z}}
\nc\Zo{{\Zs[z]^0}}
\nc\gr{{\on{gr}}}

\nc\fD{{\frak D}}

\newcommand{\Ver}{\mathsf{V}}

\newcommand{\matC}{\mathbb{C}}

\newcommand{\matQ}{\mathbb{Q}}

\newcommand{\qm}{\mathsf{QM}}

\newcommand{\bv}{\mathsf{v}}

\usepackage{tikz}

\usetikzlibrary{decorations}
\usetikzlibrary{decorations.pathmorphing}
\usetikzlibrary{calc}

\newsavebox{\Ipm}
\savebox{\Ipm}{%
\begin{tikzpicture}[baseline= {($(current bounding box.base)-(0pt,-20pt)$)}]
\begin{scope}
\draw[dashed,line width=0.35mm] (0,0)-- (10,0) ;
\node[circle,draw,minimum size=4mm,fill=black] (c) at (0,0){};
\node[circle,draw,minimum size=4mm,fill=black] (c) at (2,0){};
\node[circle,draw,minimum size=4mm,fill=black] (c) at (4,0){};
\node[circle,draw,minimum size=4mm,fill=black] (c) at (6,0){};
\node[circle,draw,minimum size=4mm,fill=black] (c) at (8,0){};
\node[circle,draw,minimum size=4mm,fill=black] (c) at (10,0){};
\draw [-to,line width=0.5mm](0,0) -- (1.8,0);
\draw [-to,line width=0.5mm](8,0) -- (9.8,0);
\node[rectangle,draw,minimum size=4mm,fill=black] (c) at (4,-2){};
\draw [-to,line width=0.5mm](4,-2) -- (4,-0.2);
\node[rectangle,draw,minimum size=4mm,fill=black] (c) at (6,-2){};
\draw [-to,line width=0.5mm](6,-2) -- (6,-0.2);
\node at (0,0.5) { $\bv_1$};
\node at (2,0.5) { $\bv_2$};
\node at (4,0.5) { $\bv_k$};
\node at (6,0.5) { $\bv_{n-k}$};
\node at (8,0.5) { $\bv_{n-2}$};
\node at (10,0.5) { $\bv_{n-1}$};
\node at (4,-2.6) { $z_1$};
\node at (6,-2.6) { $z_2$};
\end{scope}
\end{tikzpicture}}

\newsavebox{\xbox}
\savebox{\xbox}{%
\begin{tikzpicture}[baseline= {($(current bounding box.base)-(0pt,-30pt)$)}]
\begin{scope}
\draw[line width=0.35mm] (0,0)-- (5,5) ;

\draw[line width=0.35mm] (-1,1)-- (4,6) ;

\draw[line width=0.35mm] (-2,2)-- (3,7) ;

\draw[line width=0.35mm] (-3,3)-- (2,8) ;

\draw[line width=0.35mm] (0,0)-- (-3,3) ;
\draw[line width=0.35mm] (1,1)-- (-2,4) ;
\draw[line width=0.35mm] (2,2)-- (-1,5) ;
\draw[line width=0.35mm] (3,3)-- (0,6) ;
\draw[line width=0.35mm] (4,4)-- (1,7) ;
\draw[line width=0.35mm] (5,5)-- (2,8) ;

\node at (-2,3) { $x_{1,1}$};

\node at (-1,2) { $x_{2,1}$};
\node at (-1,4) { $x_{2,2}$};

\node at (0,1) { $x_{3,1}$};
\node at (0,3) { $x_{3,2}$};
\node at (0,5) { $x_{3,3}$};

\node at (1,2) { $x_{4,1}$};
\node at (1,4) { $x_{4,2}$};
\node at (1,6) { $x_{4,3}$};

\node at (2,3) { $x_{5,1}$};
\node at (2,5) { $x_{5,2}$};
\node at (2,7) { $x_{5,3}$};

\node at (3,4) { $x_{6,1}$};
\node at (3,6) { $x_{6,2}$};

\node at (4,5) { $x_{7,1}$};
\end{scope}
\end{tikzpicture}}

\newsavebox{\gbox}
\savebox{\gbox}{%
\begin{tikzpicture}[baseline= {($(current bounding box.base)-(0pt,-30pt)$)}]
\begin{scope}
\draw[line width=0.35mm] (0,0)-- (5,5) ;

\draw[line width=0.35mm] (-1,1)-- (4,6) ;

\draw[line width=0.35mm] (-2,2)-- (3,7) ;

\draw[line width=0.35mm] (-3,3)-- (2,8) ;

\draw[line width=0.35mm] (0,0)-- (-3,3) ;
\draw[line width=0.35mm] (1,1)-- (-2,4) ;
\draw[line width=0.35mm] (2,2)-- (-1,5) ;
\draw[line width=0.35mm] (3,3)-- (0,6) ;
\draw[line width=0.35mm] (4,4)-- (1,7) ;
\draw[line width=0.35mm] (5,5)-- (2,8) ;

\draw [-to,line width=0.7mm](1,2) -- (0+0.05,1+0.05);
\draw [-to,line width=0.7mm](2,3) -- (1+0.05,2+0.05);
\draw [-to,line width=0.7mm](3,4) -- (2+0.05,3+0.05);
\draw [-to,line width=0.7mm](4,5) -- (3+0.05,4+0.05);

\draw [-to,line width=0.7mm](0,3) -- (-1+0.05,2+0.05);
\draw [-to,line width=0.7mm](1,4) -- (0+0.05,3+0.05);
\draw [-to,line width=0.7mm](2,5) -- (1+0.05,4+0.05);
\draw [-to,line width=0.7mm](3,6) -- (2+0.05,5+0.05);

\draw [-to,line width=0.7mm](-1,4) -- (-2+0.05,3+0.05);
\draw [-to,line width=0.7mm](0,5) -- (-1+0.05,4+0.05) ;
\draw [-to,line width=0.7mm](1,6) -- (0+0.05,5+0.05);
\draw [-to,line width=0.7mm](2,7) -- (1+0.05,6+0.05);

\draw [-to,line width=0.7mm](-1,2) -- (0-0.05,1+0.05);
\draw [-to,line width=0.7mm](-2,3) --(-1-0.05,2+0.05) ;

\draw [-to,line width=0.7mm](0,3) -- (1-0.05,2+0.05);
\draw [-to,line width=0.7mm](-1,4) -- (0-0.05,3+0.05);

\draw [-to,line width=0.7mm](1,4) -- (2-0.05,3+0.05);
\draw [-to,line width=0.7mm](0,5) -- (1-0.05,4+0.05);

\draw [-to,line width=0.7mm](2,5) -- (3-0.05,4+0.05);
\draw [-to,line width=0.7mm](1,6) -- (2-0.05,5+0.05);

\draw [-to,line width=0.7mm](3,6) -- (4-0.05,5+0.05);
\draw [-to,line width=0.7mm](2,7) -- (3-0.05,6+0.05);

\draw [-to,line width=0.7mm](0,1) -- (0,-0.5);

\draw [-to,line width=0.7mm](2,8.5) -- (2,7);

\node at (0,-1) { $z_{1}$};

\node at (2,9) { ${z_{2}}$};

\end{scope}
\end{tikzpicture}}

\newsavebox{\hweights}
\savebox{\hweights}{%
\begin{tikzpicture}[baseline= {($(current bounding box.base)-(0pt,-30pt)$)}]
\begin{scope}
\draw[line width=0.35mm] (0,0)-- (5,5) ;

\draw[line width=0.35mm] (-1,1)-- (4,6) ;

\draw[line width=0.35mm] (-2,2)-- (3,7) ;

\draw[line width=0.35mm] (-3,3)-- (2,8) ;

\draw[line width=0.35mm] (0,0)-- (-3,3) ;
\draw[line width=0.35mm] (1,1)-- (-2,4) ;
\draw[line width=0.35mm] (2,2)-- (-1,5) ;
\draw[line width=0.35mm] (3,3)-- (0,6) ;
\draw[line width=0.35mm] (4,4)-- (1,7) ;
\draw[line width=0.35mm] (5,5)-- (2,8) ;
\node at (0,1) { $1$};
\node at (-1,2) { $2$}; \node at (1,2) { $2$};
\node at (-2,3) { $3$};\node at (0,3) { $3$};
\node at (2,3) { $3$};
\node at (-1,4) { $4$};
\node at (1,4) { $4$};
\node at (3,4) { $4$};
\node at (0,5) { $5$};
\node at (2,5) { $5$};
\node at (4,5) { $5$};
\node at (1,6) { $6$};
\node at (3,6) { $6$};
\node at (2,7) { $7$};
\end{scope}
\end{tikzpicture}}

\begin{document}

\hrule width0pt
\vsk->

\title[Polynomial superpotential for  $\Gr(k,n)$ from a limit of vertex functions]
{Polynomial superpotential for 
Grassmannian $\Gr(k,n)$ from a limit of vertex function}

\author
[Andrey Smirnov and Alexander Varchenko]
{Andrey Smirnov$^{\star}$ and Alexander Varchenko$^{\diamond}$}

\maketitle

\begin{center}
{ Department of Mathematics, University
of North Carolina at Chapel Hill\\ Chapel Hill, NC 27599-3250, USA\/}

\end{center}

\vsk>
{\leftskip3pc \rightskip\leftskip \parindent0pt \Small
{\it Key words\/}:  Superpotentials; Vertex Functions; $J$-functions; Landau-Ginzburg model.

\vsk.6>
{\it 2020 Mathematics Subject Classification\/}: 14G33 (11D79, 32G34, 33C05, 33E30)
\par}


{\let\thefootnote\relax
\footnotetext{\vsk-.8>\noindent
$^\star\<${\sl E\>-mail}:\enspace asmirnov@email.unc.edu
\\
$^\diamond\<${\sl E\>-mail}:\enspace  anv@email.unc.edu}}

\begin{abstract}

In this note we discuss an integral representation for the vertex function 
of the cotangent bundle over the Grassmannian, 
 $X=T^{*}\Gr(k,n)$. This integral representation can be used to compute
 the $\hbar\to \infty$ limit of the vertex function, where $\hbar$ denotes the equivariant parameter of a torus acting on $X$ by dilating the cotangent fibers. 
 We show that in this limit the integral turns into the standard mirror integral representation for the $A$-series 
 of the Grassmannian $\Gr(k,n)$ with the Laurent polynomial Landau-Ginzburg superpotential of Eguchi, Hori and Xiong.
We also observe some Dwork type congruences for the coefficients of the $A$-series.

\end{abstract}


\setcounter{footnote}{0}
\renewcommand{\thefootnote}{\arabic{footnote}}

\section{Introduction}

\subsection{}
The vertex functions have been introduced in \cite{Oko17} as generating 
functions counting rational quasimaps to Nakajima varieties.  In this respect, a vertex function 
is a ``quasimap'' analog of Givental's $J$-function in quantum cohomology. 
In this paper we consider the the cohomological vertex function for the cotangent
 bundle over the Grassmanian $X=T^{*}\Gr(k,n)$. By definition, this function 
 is a power series in the quantum parameter $z$ with coefficients in the equivariant cohomology:
$$
{\textsf{Vertex}}(z) \in H^{\bullet}_{T}(X)[[z]]
$$
where $T$ is a torus acting on $X$, see Section \ref{vertexfun}. Let $\Ver(z)$ denote the coefficient of the fundamental class in the vertex function
$$
\Ver(z):=\Big\langle {\textsf{Vertex}}(z), [X] \Big\rangle
$$
where $\langle - , - \rangle$ stands for the standard pairing
 in the equivariant cohomology. The power series $\Ver(z)$ is the analog of the so-called $A$-series in quantum cohomology. 
The coefficients of $\Ver(z)$ depend non-trivially on the equivariant parameter
 $\hbar$, which corresponds to the torus acting on $X$ by dilating the cotangent fibers. In this note we describe the following result (Theorem \ref{mainthm}):
\begin{thm} \label{themintro}
In the non-equivariant specialization one has the following limit:
\bean \label{mresult}
\lim\limits_{\hbar \to \infty}\, \Ver(z/\hbar^{n}) =   \dfrac{1}{(2 \pi \sqrt{-1})^{k(n-k)}} \oint e^{\frac{1}{\epsilon} S({x},{z}) }\,  \bigwedge\limits_{i,j} \frac{dx_{i,j}}{x_{i,j}}
\eean
where $S({x},{z})$ denotes a Laurent polynomial in $k(n-k)$ variables $x=(x_{i,j})$ given by (\ref{polsuper}).
\end{thm}
 The integral in the right-hand
  side of (\ref{mresult}) denotes the constant term of the integrand, 
  see Section \ref{expintdef} for the definition. 
The Laurent polynomial $S({x},{z})$ appearing in the limit above is the
 well-known version of the
 {\it superpotential} for $\Gr(k,n)$. It first appeared in the paper \cite{EHX} by 
 Eguchi-Hori-Xiong  
 and has since been revisited and generalized by many authors in the extensive literature on 
 quantum cohomology of Grassmannians. A partial listing of references can be found in 
 \cite{BCKS98,Cas,GGI,KoSt,LaTe,MaRi}. 
  Therefore, the expression on the right-hand side of equation \eqref{mresult}
   corresponds to the widely used integral representation for the $A$-series of the
   Grassmannian $\Gr(k,n)$.
  A closed combinatorial formula for 
 the coefficients of this series is also known, it was first conjectured in \cite{BCKS98} and later proved in  \cite{MaRi}, see Corollary 4.8 in \cite{MaRi} (our $z$ is 
 their $q$).

Informally speaking formula (\ref{mresult}) means that in the limit $\hbar\to\infty$ 
the cotangent directions of $X$ do not contribute to the quasimap partition function
 and the vertex function of $T^{*}\Gr(k,n)$ degenerates to the $J$-function of $\Gr(k,n)$. 
 The idea that $\hbar\to\infty$ bridges the vertex functions with the $J$-functions is not new, 
see \cite{TV22} or  Section 5 of \cite{KPSZ}.  However, the derivation of the
Laurent superpotential of \cite{EHX} from the limit of $\Ver(z)$ has not been 
documented well. The goal of this letter is to fill up this gap in the existing literature. 
 The main technical tool which allows us to compute the limit is the integral 
 representation for $\Ver(z)$ obtained previously by the authors in
Theorem 3.2 of \cite{SmV23}.

\subsection{}
As an illustration, let us consider the statement of Theorem \ref{themintro} in the simplest case. 

\medskip
\noindent
{\bf Example.}
Let $X=T^{*} \mathbb{P}(\mathbb{C}^2)$. There is a two-dimensional torus 
$\mathsf{A}=(\mathbb{C}^{\times})^{2}$ 
with equivariant parameters $u_1, u_2$
naturally acting on $\mathbb{C}^2$  by dilating the coordinate subspaces. 
There is also a one-dimensional torus $\mathbb{C}^{\times}_{\hbar}$
 acting on $X$ by dilation of the cotangent fibers, we denote the
  corresponding character by $\hbar$. Finally, there is a one-dimensional torus $\mathbb{C}^{\times}_{\epsilon}$
  with equivariant parameter $\epsilon$ 
 which acts on the source of the quasimaps  $\mathbb{P}^1 \dashrightarrow X$. 
  By definition, the vertex function is a power series
   with coefficients in equivariant cohomology \cite{Oko17}:
$$
\textsf{Vertex}(z) \in H^{\bullet}_{\mathsf{A}\times \mathbb{C}^{\times}_{\hbar} \times \mathbb{C}^{\times}_{\epsilon} }(X)_{loc} [[z]]
$$
where the subscript $loc$ denotes localization with respect to $\mathbb{C}^{\times}_{\epsilon}$. In the basis of \\
 $H^{\bullet}_{\mathsf{A}\times \mathbb{C}^{\times}_{\hbar} \times \mathbb{C}^{\times}_{\epsilon} }(X)_{loc}$ given by the classes of the torus fixed points $[p_1],[p_2] \in X^{\mathsf{A}}$
  (which correspond to the coordinate lines in $\mathbb{C}^2$)
   we have the following closed formulas in terms of the Gauss hypergeometric functions:
\bean \label{hyperfunc}
\begin{array}{c}
\Big\langle \textsf{Vertex}(z), [p_1] \Big\rangle = {}_2 F_{1}\Big(\dfrac{\hbar}{\epsilon}, \dfrac{u_2-u_1+\hbar}{\epsilon}; \dfrac{u_2-u_1+\epsilon}{\epsilon}; z \Big), \\ 
\\
\Big\langle \textsf{Vertex}(z), [p_2] \Big\rangle = {}_2 F_{1}\Big(\dfrac{\hbar}{\epsilon}, \dfrac{u_1-u_2+\hbar}{\epsilon}; \dfrac{u_1-u_2+\epsilon}{\epsilon}; z \Big).
\end{array}
\eean
In the  {\it non-equivariant} limit $u_1=u_2=0$, corresponding  to ``turning off'' 
the action of the torus $\mathsf{A}$, the above functions coincide and give 
the coefficient of the vertex function at the fundamental class~$[X]$:
\bean \label{limhyperg}
\Big\langle \textsf{Vertex}(z), [X] \Big\rangle ={}_2 F_{1}\Big(\dfrac{\hbar}{\epsilon}, \dfrac{\hbar}{\epsilon}; 1; z \Big)
\eean
We denote this coefficient by $\Ver(z)$. Explicitly we have
$$
\Ver(z)=\sum\limits_{d=0}^{\infty}\,  \dfrac{(\hbar)_{d}^2}{(d!)^2 \epsilon^{2 d}} \, z^d, \ \ \ \mathrm{where} \ \ \  (\hbar)_d=\hbar (\hbar+\epsilon)(\hbar+2\epsilon)\dots (\hbar+(d-1) \epsilon).
$$
Since
$
\lim\limits_{\hbar\to \infty}\, {(\hbar)_{d}}/{\hbar^d} =1$,\  we obtain:
$$
\lim\limits_{\hbar\to \infty}\, \Ver(z/\hbar^{2})=\sum\limits_{d=0}^{\infty}\,  \dfrac{z^d }{(d!)^2 \epsilon^{2d}}.
$$
Let $S(x,z)=x+z/x$, then
$$
\oint \dfrac{dx}{x} \,  S(x,z)^{d}:= [S(x,z)^{d}]_{0}= 
\left\{\begin{array}{ll} \dfrac{(d)!}{(d/2)! (d/2)!} z^d, & d \ \ \textrm{is even} \\ \\
0, & d \ \ \textrm{is odd}
\end{array}\right.
$$
where $[S(x,z)^{d}]_{0}$ denotes the constant term in $x$ of
 the Laurent polynomial $S(x,z)^{d}$. Combining all these together,  we can write
$$
\lim\limits_{\hbar\to \infty}\, \Ver(z/\hbar^{2})= 
\sum_{d=0}^{\infty}\, \frac{1}{{d!}} \oint \dfrac{dx}{x}  {S(x,z)^d} = \oint \dfrac{dx}{x} \, e^{\frac{S(x,z)}{\epsilon}}.
$$
This formula  is the statement of Theorem \ref{mresult} in this case. 
A more straightforward way to compute this limit is to note that the hypergeometric function
 (\ref{limhyperg}) has the integral representation:
\bean \label{inrepex}
\Ver(z)=\oint\limits_{|x|=\varepsilon}\, \dfrac{dx}{x} \, \Big(1-x\Big)^{-\frac{\hbar}{\epsilon}} \Big(1-\frac{z}{x}\Big)^{-\frac{\hbar}{\epsilon}}
\eean
where $\varepsilon$ is any positive real number such that $|z|< \varepsilon < 1$. We note that 
the change of variables $z\to z/\hbar^{2}$, $x\to x/\hbar$ together with
change of the contour $\varepsilon \to \varepsilon/\hbar$ does not affect 
this condition for large $|\hbar|$. Thus, for large $|\hbar|$ we have:
$$
\Ver(z/\hbar^{2})=\oint\limits_{|x|=\varepsilon}\, \dfrac{dx}{x} \, \Big(1-\frac{x}{\hbar}\Big)^{-\frac{\hbar}{\epsilon}} \Big(1-\frac{z}{x \hbar}\Big)^{-\frac{\hbar}{\epsilon}}
$$
which allows one to compute the limit using elementary tools:
$$
\lim\limits_{\hbar \to \infty}\, \Big(1-\frac{x}{\hbar}\Big)^{-\frac{\hbar}{\epsilon}} \Big(1-\frac{z}{x \hbar}\Big)^{-\frac{\hbar}{\epsilon}} = e^{\frac{S(x,z)}{\epsilon}}
$$
\subsection{Exposition of the material} 
In Section \ref{vertexfun} we recall a combinatorial formula for the vertex functions 
generalizing (\ref{hyperfunc}) to the case of $X=T^{*}\Gr(k,n)$. In Section \ref{mirrsymsect} we describe the analog of the integral representation (\ref{inrepex}) for this case. In Section~\ref{explimsec} we use this integral representation to compute the $\hbar\to\infty$ limit of $\Ver(z)$ similarly to  the example above. 

In our previous paper \cite{SmV23},  we show that certain truncations of  $\Ver(z)$  with parameters specialized to $\mathbb{Q}_p$ satisfy some Dwork type congruence relations. 
In Section~\ref{DworkCong} we show that a similar structure exists  in the limit 
$\hbar\to \infty$.

\subsection*{Acknowledgements}
We thank Thomas Lam for very useful comments. 
Work of A. Smirnov is partially supported by NSF grant DMS - 2054527 and by
the RSF under grant 19-11-00062.
Work of A. Varchenko is partially supported by NSF grant DMS - 1954266.

\section{The vertex function of $T^{*}\Gr(k,n)$ \label{vertexfun}}

\subsection{}

For $X=T^{*}\Gr(k,n)$ we consider the following explicit power series:
\bean \label{verfun}
\langle \textsf{Vertex}(z),[1,\dots, k] \rangle: =\sum_{d=0}^{\infty} \, c_{d}(u_1,\dots, u_n,\hbar) \, z^d 
\eean
with the coefficients $c_{d}(u_1,\dots, u_n,\hbar)\in \matQ(u_1,\dots,u_n,\hbar,\epsilon)$ given by:
\begin{small}
\bean \label{coeffsd}
 c_{d}(u_1,\dots, u_n,\hbar)=\sum\limits_{{d_1,\dots,d_k:}\atop {d_1+\dots+d_k=d}}\,  \Big(\prod\limits_{i,j=1}^{k} \dfrac{(\epsilon-u_i+u_j)_{d_i-d_j}}{(\hbar-u_i+u_j)_{d_i-d_j}} \Big) \Big(\prod\limits_{j=1}^{n} \prod\limits_{i=1}^{k}\, \dfrac{(\hbar+u_j-u_i)_{d_i}}{(\epsilon+u_j-u_i)_{d_i}} \Big),
\eean
\end{small}
where $(x)_d$ denotes the Pochhammer symbol with step $\epsilon$:
$$
(x)_d =\left\{ \begin{array}{rr}
x (x+\epsilon)\dots (x+(d-1) \epsilon), & d>0\\
1, & d=0\\
\dfrac{1}{(x-\epsilon)(x-2 \epsilon) \dots (x +d \epsilon)  }, & d<0
\end{array}\right. 
$$
The degree $d$ coefficient of this series counts (equivariantly) the number of degree $d$ rational curves in $X$. More precisely, it is given by the equivariant integral
\bean \label{cinteg}
c_{d}(u_1,\dots, u_n,\hbar) = \int\limits_{ [\qm_d(X,\infty)]^{\textrm{vir}}} \, \omega^{vir}
\eean 
over the virtual fundamental class on moduli space $\qm_d(X,\infty)$ of quasimaps from $\mathbb{P}^1$ to $X$, which send $\infty \in \mathbb{P}^1$ to a prescribed torus fixed point $[1,\dots, k]\in X$, see Section 7.2 of \cite{Oko17} for definitions.
 Using the equivariant localization,  the integral (\ref{cinteg}) reduces to the sum over the torus fixed points on $\qm_d(X,\infty)$ which gives the sum (\ref{coeffsd}). We refer to Section 4.5 of \cite{PSZ16} where this computation is done in some details.

The parameters $u_1,\dots, u_n,\hbar,\epsilon$ are the equivariant parameters of the torus $T=(\mathbb{C}^{\times})^{n}\times \mathbb{C}^{\times}_{\hbar} \times  \mathbb{C}^{\times}_{\epsilon}$ acting on the moduli space $\qm_d(X,\infty)$ in the following way: 
\begin{itemize}
    \item $(\mathbb{C}^{\times})^{n}$ acts on $\mathbb{C}^n$ in a natural way, scaling the coordinates with weights $u_1,\dots , u_n$. 

    \item The set of torus fixes points $X^{(\mathbb{C}^{\times})^{n}}$ corresponds
to $k$-subspaces in $\mathbb{C}^n$ spanned by any set of $k$ coordinate lines. The fixed point $[1,\dots, k] \in X^{(\mathbb{C}^{\times})^{n}}$  corresponds to the $k$-subspace spanned by the first $k$ coordinate lines. 
    
    \item $\mathbb{C}^{\times}_{\hbar}$ acts on $X$ by scaling the cotangent fibers with weight $\hbar$.
    
    \item $\mathbb{C}^{\times}_{\epsilon}$ acts on the source of the quasimaps $C\cong\mathbb{P}^{1}$ 
     fixing the points $0,\infty \in \mathbb{P}^{1}$. The parameter $\epsilon$ denotes the corresponding weight of the tangent space $T_{0}\, C$.
\end{itemize}

The full vertex function is a power series 
with coefficients in equivariant cohomology:
$$
\textsf{Vertex}(z) \in H^{\bullet}_{T} (X)_{loc}[[z]]
$$
where $loc$ denotes the equivariant localization with respect to torus $\mathbb{C}^{\times}_{\epsilon}$.
Using the equivariant localization, we can expand $\textsf{Vertex}(z)$ in the basis of $H^{\bullet}_{T\times \mathbb{C}^{\times}_{\epsilon}} (X)_{loc}$ given by the classes of torus fixed points. The power series (\ref{verfun}) gives the coefficient $\textsf{Vertex}(z)$ 
at the ``first'' torus fixed point $[1,\dots,k]$.  Other coefficients have the same structure and can be obtained from (\ref{verfun}) by permutations of parameters $u_i$.

\subsection{} 

In this paper we consider the specialization of the equivariant parameters:
\bean \label{specializ}
&& 
u_1=0, \dots, u_n =0
\eean
which corresponds to non-equivariant limit when the action of the torus $(\mathbb{C}^{\times})^{n}$ is ``turned off''. The coefficients of $\textsf{Vertex}(z)$ at the torus fixed points all reduce to the same function
(simply because without $(\mathbb{C}^{\times})^{n}$-action these points are indistinguishable) which corresponds to the coefficient of the vertex function at the fundamental class:
\bean \label{verfundefn}
\Ver(z):=\left.\Big\langle \textsf{Vertex}(z), [X] \Big\rangle\right|_{u_1=0,\dots, u_n=0} 
\eean 
Thus, $\Ver(z)$ can be obtained by specializing the coefficients of the power series (\ref{verfun}) at (\ref{specializ}). We note that this specialization is non-trivial: already in the case of $T^{*}\Gr(2,4)$ the terms in the sum (\ref{coeffsd})
have poles at $u_i=u_j$. The total sum (\ref{coeffsd}) is, however, non-singular  
since the vertex function is an integral equivariant cohomology class (we recall that only $\mathbb{C}^{\times}_{\epsilon}$ - localization is required to define it).

\section{Integral representation of $\Ver(z)$ \label{mirrsymsect}}

\subsection{} 
In this section we describe an integral representation for the function (\ref{verfundefn})
$$
\Ver(z) = \int_{\gamma}\, \Phi(x,z)\, dx
$$
which has its origin in $3D$-mirror symmetry, we refer to
Section 3 of \cite{SmV23} for more details.

\subsection{} 
Assume that $n\geq 2k$. Let $\bv_i$, $i=1,\dots, n-1$ be integers defined by:
$$
\bv_i=\left\{\begin{array}{ll}
i, & i< k,\\
k, & k \leq i \leq n-k,\\
n-i, & n-k < i,
\end{array}\right. 
$$
We denote by $\om=\hbar/\epsilon$ and define the {\it superpotential} function:
\bean 
\label{superpot}
\Phi(x,z)
&=&
\Big(\prod\limits_{i=1}^{n-1} \prod\limits_{j=1}^{\bv_i}\, x_{i,j} \Big)^{-1+\om} 
\Big(\prod\limits_{m=1}^{\bv_m}\, \prod\limits_{1\leq i<j \leq \bv_m}  (x_{m,j}-x_{m,i})\Big)^{2\om}
\\
\notag
&\times&
\Big(\prod\limits_{i=1}^{n-2} \prod\limits_{a=1}^{\bv_i} \prod\limits_{b=1}^{\bv_{i+1}} (x_{i,a}-x_{i+1,b})\Big)^{-\om} 
\Big(\prod\limits_{i=1}^{k} (z_{1}-x_{k,i}) (z_{2}- x_{n-k,i})\Big)^{-\om} .
\eean
We note that this function is an example of the {\it master functions} in the theory of integral representations of the trigonometric Knizhnik-Zamolodchikov equations.  In particular, (\ref{superpot}) corresponds to
the KZ equation associated with the weight subspace of weight [1,\dots,1] in the tensor product the $k$-th and $(n-k)$-th fundamental representations of $\frak{gl}_n$, see \cite{SV91,MV02}.

\subsection{}
The dimension vector $\bv_i$ and the variables $x_{i,j}$ have a convenient combinatorial visualization. Let us consider a $k\times (n-k)$ rectangle rotated counterclockwise by $45^{\circ}$, see Fig. \ref{xboxref}. Note that in this picture the number of boxes in $i$-th vertical column is exactly $\bv_i$. In this way, we may assign the variables $x_{i,j}$ to the boxes in this picture. We will order them as in 
Fig. \ref{xboxref}. 
\begin{figure}[h!]
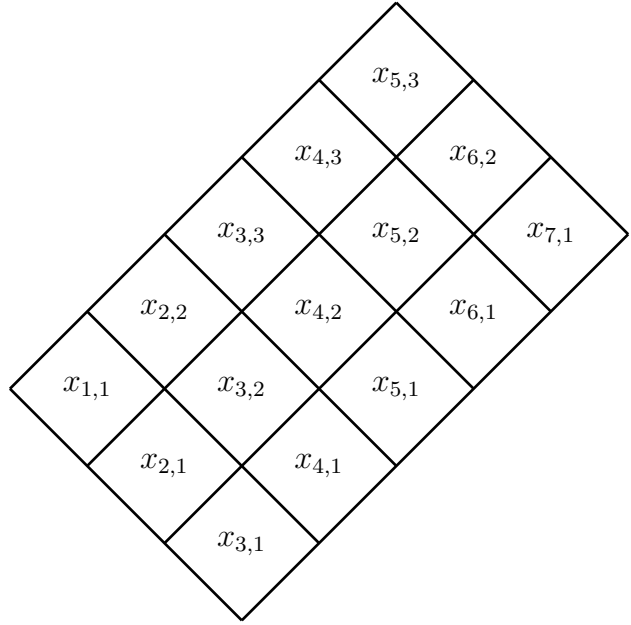

$$
\resizebox{0.5\textwidth}{!}{%
\usebox{\xbox}}
$$
\caption{Set of variables $x_{i,j}$ for $k=4$ and $n=8$. \label{xboxref}} 
\end{figure} Note that the total number of variables $x_{i,j}$ equals to $\dim \Gr(k,n) = k(n-k)$. 
To a box $(i,j)$ in the Fig. \ref{xboxref} we assign a weight 
\bean \label{weightfundef}
m_{i,j}= (|i-k|+2 j-1) \in \mathbb{N}
\eean
This function ranges from $m_{k,1}=1$ to $m_{n-k,k}=n-1$. The definition of $m_{i,j}$ is clear from the Fig.\ref{mweightpic}.
We have a partial ordering on the boxes $(i,j)$ corresponding to:
\bean
 \label{ordere}
m_{k,1}< m_{k-1,1} =  m_{k+1,1} < \dots  <m_{n-k,k}
\eean 
\begin{figure}[h!]
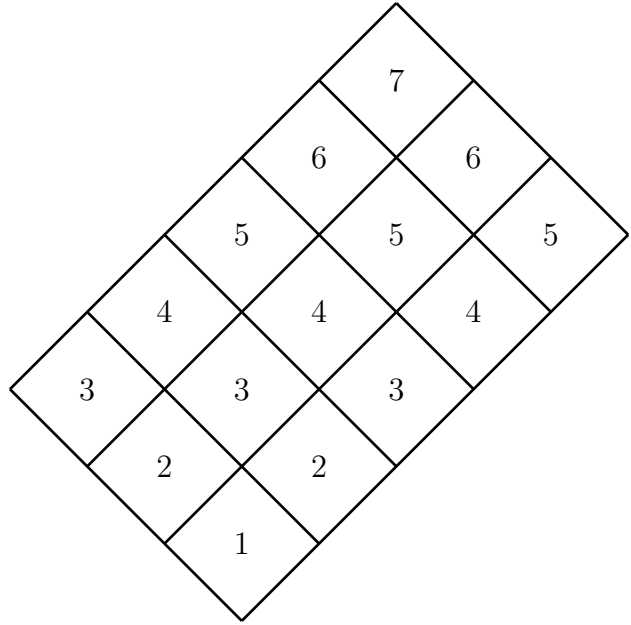

$$
\resizebox{0.5\textwidth}{!}{%
\usebox{\hweights}}
$$
\caption{The values of the weight function $m_{i,j} \label{mweightpic}$} 
\end{figure}
For a small real number $0<\varepsilon\ll 1$ let us define  the torus  by the following equations:
\bean \label{conctour}
\gamma_{k,n} \subset \matC^{k(n-k)}, \ \ \ |x_{i,j}|= m_{i,j} \varepsilon
\eean
where $i, j$ run through all possible values.

\begin{prop} [\cite{SmV23}]
 \label{branchprop} 

 Assume that $|z_{1}|<\varepsilon$ and $(n-1) \varepsilon <|z_{2}|$, then  the superpotential (\ref{superpot}) has 
a single-valued branch on the torus $\gamma_{k,n}$, which is distinguished in the proof and which 
will be used in the paper.

\end{prop}

\begin{proof}
 Let us denote 
 \bean \label{Lfactor}
 L(x_{i,a},x_{j,b})=\left\{\begin{array}{ll}
 (1-x_{i,a}/x_{j,b})^{-\om}, & m_{i,a} < m_{j,b}, \  i\ne j\\
(x_{j,b}/x_{i,a}-1)^{-\om}, & m_{i,a} > m_{j,b},  \   i\ne j.
 \end{array}\right.
 \eean
Each of these 
ratios $x_{i,a}/x_{j,b}$,
$x_{j,b}/x_{i,a}$ restricted to $\gamma_{k,n}$
has absolute value less than 1.
We replace  $ (1-x_{i,a}/x_{j,b})^{-\om}$ on $\gamma_{k,n}$
with 
$\sum_{m=0}^\infty\binom{-\om}{m}(-x_{i,a}/x_{j,b})^m$ and replace
\\
$(x_{j,b}/x_{i,a}-1)^{-\om}$ with 
$e^{-\pi \sqrt{-1}\om} \sum_{m=0}^\infty\binom{-\om}{m}(-x_{j,b}/x_{i,a})^m$.

Next, we denote 
$L(z_{1},x_{k,a})= (1-z_{1}/x_{k,a})^{-\om}$ and 
$L(z_{2},x_{k,a})=
\\
 (1-x_{n-k,a}/z_{2})^{-\om }$.
On  $\gamma_{k,n}$ we have $|x_{k,i}|\geq \epsilon$, and $|x_{n-k,i}|\leq |x_{n-k,k}|= n \epsilon$, therefore
$|z_{1}/x_{k,i}|<1$ and $|x_{n-k,i}/z_{2}|<1$. 
We replace on $\gamma_{k,n}$ the factor
$(1-z_{1}/x_{k,a})^{-\om }$ with 
$\sum_{m=0}^\infty\binom{-\om}{m}(-z_{1}/x_{k,a})^m$ and the factor 
$(1-x_{n-k,a}/z_{2})^{-\om}$ with
$\sum_{m=0}^\infty\binom{-\om}{m}(-x_{n-k,a}/z_{2})^m$.

Finally, we denote $\Delta(x_{m,i},x_{m,j})=
 (1-x_{m,i}/x_{m,j})^{2\om }$ for $1\leq i<j\leq \bv_m$.
On $\ga_{k,n}$ we have $|x_{m,i}/x_{m,j}|<1$.
We replace on $\gamma_{k,n}$ the factor
$\Delta(x_{m,i},x_{m,j})$ with 
$\sum_{m=0}^\infty\binom{2\om}{m}(-x_{m,i}/x_{m,j})^m$.

In these notations we have:
 \bean \label{Lfacts}
 \\
 \notag
 \Phi(x,z)=  \frac{\Big(\prod\limits_{i=1}^{n-1} \prod\limits_{a<b}\Delta(x_{i,a},x_{i,b})\Big)\,
\Big(\prod\limits_{i=1}^{n-2} \prod\limits_{a=1}^{\bv_i} \prod\limits_{b=1}^{\bv_{i+1}} L(x_{i,a},x_{i+1,b})\Big) 
\Big(\prod\limits_{i=1}^{k} L(z_{1},x_{k,i})  L(z_{2},x_{n-k,i})\Big)}{ \prod\limits_{i=1}^{n-1}  \prod\limits_{j=1}^{\bv_i} x_{i,j}  } ,
 \eean
and for each factor a single-valued branch is chosen
 by replacing that factor with the corresponding power series. The product of those power series distinguishes a single-valued branch of 
 $\Phi(x,z)$ on $\gamma_{k,n}$.
\end{proof}

\begin{example} For $X=T^*\Gr(2,4)$ we have
\bea
\Phi(x,z)
&=& (x_{11}x_{21}x_{22}x_{31})^{-1}
\\
&\times &
 (1-x_{21}/x_{22})^{2\om} 
\big( (x_{21}/x_{11}-1) (1-x_{21}/x_{31}) (z_{1}/x_{21}-1)(1-x_{21}/z_{2})\big)^{-\om}
\\
&\times& 
\big( (1-x_{11}/x_{22}) (x_{31}/x_{22}-1) (z_{1}/x_{22}-1)(1-x_{22}/z_{2})\big)^{-\om}.
\eea

\end{example}

From the previous proposition, the integral of $\Phi(x,z)$ over $\gamma_{k,n}$ is an analytic function of $z=z_{1}/z_{2}$ in the disc $|z|<\epsilon$. 
\begin{thm}[\cite{SmV23}] 
\label{verttheorem}
The  function (\ref{verfundefn})   has the following 
integral representation
\bean 
\label{vertexInt}
\Ver(z) = \dfrac{\alpha}{(2 \pi \sqrt{-1})^{k(n-k)}} \oint\limits_{\gamma_{k,n}}\, \Phi(x,z) \, \bigwedge\limits_{i,j} dx_{i,j} 
\eean 
where $\Phi(x,z)$ is the branch of superpotential function (\ref{superpot})
 on the torus $\gamma_{k,n}$ chosen in Proposition \ref{branchprop}, and $\alpha=e^{\pi\sqrt{-1}N\om}$ 
 is a normalization constant  where $N$ is the number  of factors 
in \eqref{Lfacts} having the form $(x_{j,b}/x_{i,a}-1)^{-\om}$. 
\end{thm}

\begin{defn} \label{defequiv}
 Let $\gamma^{'}_{k,m}$ be another contour defined by 
 $|x_{i,j}| = R_{i,j}$ for $R_{i,j} \in \mathbb{R}$ such that 
 $|z_{1}| < R_{1,1},  R_{n-k,k} <|z_{2}|$ and the conditions 
 \bean \label{contcond}
 m_{i,j}<m_{a,b} \Longrightarrow R_{i,j}<R_{a,b}
 \eean
 are satisfied for all pairs of indices $(i,j)$ and $(a,b)$. Then, we say that 
 $\gamma^{'}_{k,m}$ is homologous
  to $\gamma_{k,m}$ and write $\gamma^{'}_{k,m}\sim\gamma_{k,m}$.
\end{defn}

Note that (\ref{vertexInt}) remains invariant if we replace $\gamma_{k,m}$ by a homologous $\gamma^{'}_{k,m}$.  This is simply because the evaluation of the integral over $\gamma'_{k,n}$ is again
reduces to computing the residues at $x_{i,j}=0$, and the residues are computed
 in the same order as for $\gamma_{k,n}$. This implies that the result remains the same.

\subsection{Relation to $3D$-mirror symmetry}
Let us explain the origin of the superpotential function (\ref{superpot}). The factors of  (\ref{superpot}) correspond to the edges of the quiver which describes the $3D$-{\it mirror variety} $X^{!}$.
For $X=T^{*}\Gr(k,n)$ the quiver of $X^{!}$ is given in Sections 3.2-3.3 of 
\cite{SmV23}, the correspondence between the factors of (\ref{superpot}) and the edges of this quiver is also explained there.

For for general Nakajima quiver varieties, the superpotential function (\ref{superpot})
 is constructed by the same procedure if the quiver for the $3D$-mirror variety
  $X^{!}$ is known. For the Nakajima quiver varieties of type $A$,
 which include cotangent bundles over partial flag varieties as special
  cases, a conjectural description of the $3D$-mirrors was given by physicists.
   It is explained for instance in \cite{GK}.
    We expect that the results of this note and of \cite{SmV23} have 
    straightforward generalizations to those cases.

The $3D$-mirror symmetry conjecture is formulated on the level of K-theory rather than cohomology. Recall, that the quantum difference equations \cite{OS22} are the K-theoretic generalizations of quantum differential equations in quantum cohomology. The $3D$-mirror symmetry conjecture claims that the quantum difference equations for $X$ and $X^{!}$ are equivalent.  The K-theoretic vertex functions of $X$ and $X^{!}$ provide two different bases of solutions to the this common system of $q$-difference equations. For cotangent bundles over Grassmannians this conjecture was proved by Dinkins in \cite{Din20} and for full flag varieties in \cite{Din22}. For the hypertoric varieties this result is obtained in \cite{SZ22}.

An alternative definition of $3D$-mirror symmetry postulates  the equality of the elliptic stable envelopes \cite{AO21} of $X$ and $X^{!}$. This idea was first proposed in \cite{Oko18} and later examined for various cases of $X$ in \cite{RSVZ19,RSVZ21,RW,SZ22}.  It was shown in \cite{KoSm,Ko} that the elliptic stable envelope of $X$ determines the corresponding quantum difference equation of $X$ 
and vice versa. This established an equivalence between the two definitions of $3D$-mirror symmetry.

Theorem \ref{themintro} says that the mirror description of
the $J$-function for $\Gr(k,n)$ arises as a double limit of $3D$-mirror symmetry.  In the first limit one considers the cohomological limit of the K-theoretic vertex functions for $T^{*}\Gr(k,n)$. In this limit, the $3D$-mirror symmetry description of these functions \cite{Din20} degenerates to the integral representation (\ref{vertexInt}).  In the second limit $\hbar\to \infty$ we obtain Theorem \ref{themintro}.

\section{The limit $\hbar\to \infty$ \label{explimsec}} 
\subsection{Polynomial superpotential}
Let $\Gamma$ be an oriented graph, with vertices given by boxes inside the $k\times (n-k)$ Young diagram, plus two extra vertices corresponding to $z_{1}$ and $z_{2}$, see Fig.\ref{boxgraph}. The edges of the graph are defined as follows: every two adjacent boxes are connected by an edge. 
Each edge is oriented in the direction of decrease of weight function $m_{i,j}$, which is defined by
(\ref{weightfundef}). Two additional edges are  from  $x_{k,1}$ to $z_{1}$ and from $z_{2}$ to $x_{n-k,k}$,  Fig.\ref{boxgraph}. Given an edge $e$ of $\Gamma$  we denote by $h(e)$ and $t(e)$ the corresponding head and tail. We define the following Laurent polynomial:
\bean  \label{polsuper}
S({x},{z})=\sum\limits_{e \in \mathrm{edges}(\Gamma)} \, \dfrac{x_{h(e)}}{x_{t(e)}}\,.
\eean

\begin{figure}[h!]
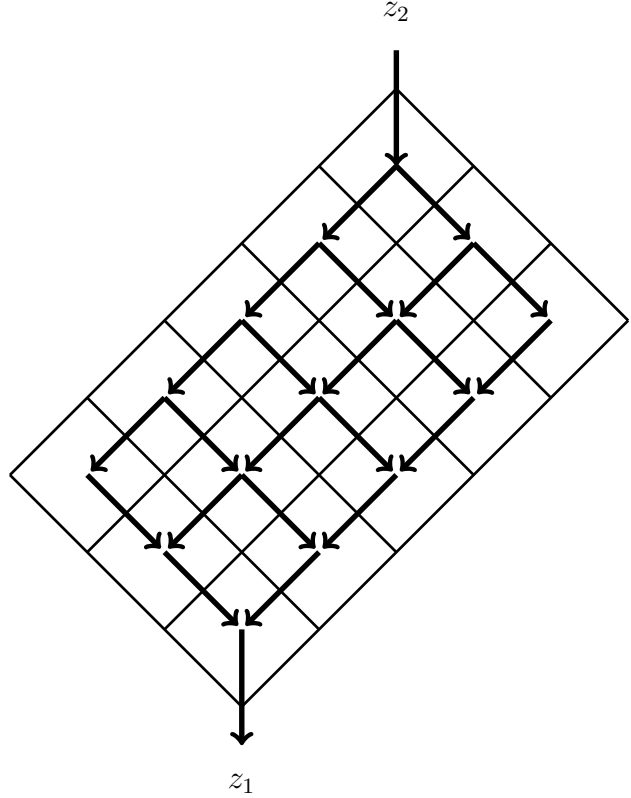

$$
\resizebox{0.5\textwidth}{!}{%
\usebox{\gbox}}
$$
\caption{The graph $\Gamma$ associated with $X=T^{*}\Gr(k,n)\label{boxgraph}.$} 
\end{figure}

\begin{example}
For $k=1$  we obtain:
\bean \label{projsuper}
S({x},{z})=
\dfrac{z_{1}}{x_{1,1}} + \dfrac{x_{1,1}}{x_{2,1}}+ \dfrac{x_{2,1}}{x_{3,1}} + \dots +  
\dfrac{x_{n-2,1}}{x_{n-1,1}} +  \dfrac{x_{n-1,1}}{z_{2}}\,.
\eean
Substituting $z_{1}=q, z_{2}=1$, and introducing new variables by 
$$
x_{1,1}=a_{1} a_2 \dots a_{n-1}, \ \ x_{2,1}=a_{1} a_2 \dots a_{n-2},  \ \ \dots, \ \ x_{n-1,1}=a_{1}\,.
$$
We arrive at the standard Givental's superpotential of projective space
$$
S(x)=a_1+a_2+\dots+a_{n-1} + \dfrac{q}{a_1\cdots a_{n-1}}\,.
$$
\end{example}

\begin{example}
For $X=T^{*}\Gr(2,4)$ we obtain
\bean \label{supgr24}
S({x},{z})=\dfrac{z_{1}}{x_{2,1}} + \dfrac{x_{2,1}}{x_{1,1}}+ \dfrac{x_{2,1}} {x_{3,1}}
+ \dfrac{x_{3,1}}{x_{2,2}} + \dfrac{x_{1,1}}{x_{2,2}} +  \dfrac{x_{2,2}}{z_{2}}\,.
\eean
\end{example}

\subsection{Exponential integral\label{expintdef}}
For $S({x},{z})$ defined by (\ref{polsuper}) we consider the power series:
\bean \label{expinteg}
\dfrac{1}{(2 \pi \sqrt{-1})^{k(n-k)}} \oint e^{\frac{1}{\epsilon} S({x},{z}) }\,  \bigwedge\limits_{i,j} \frac{dx_{i,j}}{x_{i,j}} \ \  := \ \ \sum\limits_{d=0}^{\infty} \, \dfrac{\left[S({x},{z})^d \right]_{0}}{d! \epsilon^d}\, 
\eean
where $\left[S({x},{z})^d \right]_{0}$ denotes the constant term of the Laurent polynomial $S({x},{z})^d$ in 
the variables ${x}=(x_{i,j})$. We assume that $\left[S({x},{z})^0 \right]_{0}=1$. 
From the structure of the superpotential (\ref{polsuper}) it is easy to see that 
$\left[S({x},{z})^d \right]_{0}$ is a monomial in $z=z_{1}/z_{2}$ and thus  (\ref{expinteg}) is a power series in $z$. 

\begin{example}
For $X=T^{*} \mathbb{P}^{n-1}$, the superpotential is given by (\ref{projsuper}).  In this case, elementary computations shows that 
$\left[S({x},{z})^d\right]_{0}$ is non-vanishing only if the degree $d$ is of the form 
$d=n m$ for some $m \in \mathbb{N}$. In this case
we have
$$
\left[S({x},{z})^{m n}\right]_{0} = \dfrac{(z_1/z_2)^m (n m)!}{(m!)^n}\,.
$$
We thus conclude that
$$
\oint e^{\frac{1}{\epsilon} S({x},{z}) }\,  \bigwedge\limits_{i,j} \frac{dx_{i,j}}{x_{i,j}} = \sum\limits_{m=0}^{\infty}\, \dfrac{z^m}{(m!)^n \epsilon^{m n}}\,.
$$
\end{example}

\begin{example}
For $X=T^{*}\Gr(2,4)$, the superpotential is given by (\ref{supgr24}).  In this case
$\left[S({x},{z})^d\right]_{0}$ is non-vanishing only if $d=4 m$, in which case
$$
\left[S({x},{z})^{4m}\right]_{0}= \dfrac{(2m)! (4m)!}{(m!)^6} \dfrac{z^m}{ \epsilon^{4m}}\,.
$$
Thus, we obtain:
$$
\oint e^{\frac{1}{\epsilon} S({x},{z}) }\,  \bigwedge\limits_{i,j} \frac{dx_{i,j}}{x_{i,j}} = \sum\limits_{m=0}^{\infty}\, \dfrac{(2m)!}{(m!)^6 \epsilon^{4 m}}\, z^m.
$$
\end{example}

\subsection{The vertex function in the $\hbar\to\infty$ limit}

\begin{thm} \label{mainthm}
Let $\Ver(z)$ be the function (\ref{verfundefn}), then:
$$
\lim_{\hbar \to \infty}\, \Ver(z/\hbar^{n}) =  \dfrac{1}{(2 \pi \sqrt{-1})^{k(n-k)}} \oint e^{\frac{1}{\epsilon} S({x},{z}) }\,  \bigwedge\limits_{i,j} \frac{dx_{i,j}}{x_{i,j}} 
$$
where the integral is defined by (\ref{expinteg}) and $S({x},{z})$ is the  polynomial superpotential (\ref{polsuper}).
\end{thm}
\begin{proof}
By Theorem \ref{verttheorem} we have:
$$
\Ver(z) = \dfrac{\alpha}{(2 \pi \sqrt{-1})^{k(n-k)}} \oint\limits_{\gamma_{k,n}}\, \Phi(x,z) \, \bigwedge\limits_{i,j} dx_{i,j} 
$$
where the contour $\gamma_{k,n}$ is defined by (\ref{conctour}) and $\Phi(x,z)$ is the branch of the superpotential (\ref{Lfacts}) distinguished by Proposition \ref{branchprop}. It will be convenient to define:
 \bean
 \tilde{L}(x_{i,a},x_{j,b})=\left\{\begin{array}{ll}
 (1-x_{i,a}/x_{j,b})^{-\om}, & m_{i,a} < m_{j,b}, \\
(1-x_{j,b}/x_{i,a})^{-\om}, & m_{i,a} > m_{j,b}
 \end{array}\right.
 \eean
 which differ from (\ref{Lfactor}) by a factor
 $$ 
 \tilde{L}(x_{i,a},x_{j,b})=\left\{\begin{array}{ll}
 {L}(x_{i,a},x_{j,b}), & m_{i,a} < m_{j,b}, \\
 \\
e^{-\pi \sqrt{-1} \hbar/\epsilon} {L}(x_{i,a},x_{j,b}), & m_{i,a} > m_{j,b}.
 \end{array}\right.
$$ 
Recall that $\alpha=e^{\pi \sqrt{-1} N\hbar/\epsilon}$ where $N$ is the total number of factors in $\Phi(x,z)$ for which $\tilde{L}(x_{i,a},x_{j,b})/{L}(x_{i,a},x_{j,b}) = e^{-\pi \sqrt{-1} \hbar/\epsilon}$. Thus, in these notations we have
\bean \label{redefv}
\Ver(z) = \dfrac{1}{(2 \pi \sqrt{-1})^{k(n-k)}} \oint\limits_{\gamma_{k,n}}\, \tilde{\Phi}(x,z) \, \bigwedge\limits_{i,j} \dfrac{dx_{i,j}}{x_{i,j}} 
\eean
where 
$$
\tilde{\Phi}(x,z)= \Big(\prod\limits_{i=1}^{n-1} \prod\limits_{a<b}\Delta(x_{i,a},x_{i,b})\Big)
\Big(\prod\limits_{i=1}^{n-2} \prod\limits_{a=1}^{\bv_i} \prod\limits_{b=1}^{\bv_{i+1}} \tilde{L}(x_{i,a},x_{i+1,b})\Big) 
\Big(\prod\limits_{i=1}^{k} \tilde{L}(z_{1},x_{k,i})  L(z_{2},x_{n-k,i})\Big). 
$$
In this integral we rescale the variables by:
$
z_{1} \to z_{1}, \ \ z_{2}\to z_{2} \hbar^{n},
$
since in our notations $z=z_{1}/z_{2}$, this is equivalent to substitution 
$z\to z/\hbar^{n}$ in the left-hand side of (\ref{redefv}). Let $\gamma'_{k,n}(\hbar)$ be the
 contour defined by the conditions
$$
|x_{i,j}|= m_{i,j} \varepsilon |\hbar|^{m_{i,j}}\,.
$$
Assuming that $|\hbar| > 1$ we have
$$
m_{i,j}<m_{a,b} \ \ \ \Longrightarrow \ \ \  m_{i,j}   \varepsilon |\hbar|^{m_{i,j}} < m_{a,b} \varepsilon |\hbar|^{m_{a,b}}
$$
for all pairs $(i,j)$ and $(a,b)$. By assumption of Proposition \ref{branchprop}, we have $|z_{1}|<|x_{1,1}|$ and 
 $|x_{n-k,k}| <|z_{2} \hbar^{n}|$ on $\gamma'_{k,n}(\hbar)$.  Therefore $\gamma'_{k,n}(\hbar)\sim \gamma_{k,n}$ in the sense of Definition  \ref{defequiv}. Thus

\bean \label{integrxy}
\Ver(z /\hbar^{n}) = \dfrac{1}{(2 \pi \sqrt{-1})^{k(n-k)}} \oint\limits_{\gamma^{'}_{k,n}(\hbar)}\, \tilde{\Phi}(x,z) \, \bigwedge\limits_{i,j} \dfrac{dx_{i,j}}{x_{i,j}} 
\eean
Now, in this integral we change the variables of integration by 
$
x_{i,j} = y_{i,j} \hbar^{m_{i,j}}.
$ The contour $\gamma^{'}_{k,n}(\hbar)$ in the variables $y_{i,j}$ is given by
$|y_{i,j}|=m_{i,j} \varepsilon$, 
i.e., in the coordinates $y_{i,j}$ we integrate over the original contour $\gamma_{k,n}$. Overall we obtain:
\bean \label{verchang}
\Ver(z/\hbar^{n}) = \dfrac{1}{(2 \pi \sqrt{-1})^{k(n-k)}} \oint\limits_{\gamma_{k,n}}\, \tilde{\Phi}(y,z) \, \bigwedge\limits_{i,j} \dfrac{dy_{i,j}}{y_{i,j}} 
\eean
where

\bean \nonumber
\tilde{\Phi}(y,z)= 
\Big(\prod\limits_{i=1}^{n-1} \prod\limits_{a<b}\Delta(y_{i,a} \hbar^{m_{i,a}},y_{i,b} \hbar^{m_{i,b}})\Big)\Big(\prod\limits_{i=1}^{n-2} \prod\limits_{a=1}^{\bv_i} \prod\limits_{b=1}^{\bv_{i+1}} \tilde{L}(y_{i,a} \hbar^{m_{i,a}},y_{i+1,b} \hbar^{m_{i+1,b}})\Big) \times \\ \label{phiy}
\times \Big(\prod\limits_{i=1}^{k} \tilde{L}(z_{1},y_{k,i} \hbar^{m_{k,i}})  \tilde{L}(z_{2} \hbar^{n},y_{n-k,i} \hbar^{m_{n-k,i}})\Big). 
\eean
We have:
$$ \tilde{L}(y_{i,a} \hbar^{m_{i,a}},y_{i+1,b} \hbar^{m_{i+1,b}})=\left\{\begin{array}{ll}
 \Big(1-(y_{i,a}/y_{i+1,b})/(\hbar^{m_{i+1,b}-m_{i,a}}) \Big)^{-\hbar/\epsilon}, & m_{i,a} < m_{i+1,b}, 
 \\
\Big(1-(y_{i+1,b}/y_{i,a})/(\hbar^{m_{i,a}-m_{i+1,b}})\Big)^{-\hbar/\epsilon}, & m_{i,a} > m_{i+1,b}.
 \end{array}\right.
$$
Note that the powers of $\hbar$ appearing in these factors are positive integers. 
Thus, we compute
$$
\lim\limits_{\hbar\to \infty}\, L(y_{i,a} \hbar^{m_{i,a}},y_{i+1,b} \hbar^{m_{i+1,b}})= 
\left\{\begin{array}{ll}
e^{\frac{1}{\epsilon} \frac{y_{i,a}}{y_{i+1,b}}}, & m_{i,a}=m_{i+1,b}-1,\\
\\
e^{\frac{1}{\epsilon} \frac{y_{i+1,b}}{y_{i,a}}}, & m_{i,a}=m_{i+1,b}+1, \\
\\
1, & \mathrm{otherwise.}
\end{array}\right.
$$
We also have
$$
\tilde{L}(z_{1},y_{k,i} \hbar^{m_{k,i}})=\Big(1-(z_{1}/y_{k,i})/\hbar^{m_{k,i}}\Big)^{-\hbar/\epsilon},
$$

$$
\tilde{L}(z_{2} \hbar^{n},y_{n-k,i} \hbar^{m_{n-k,i}}) = \Big(1-y_{n-k,i} \hbar^{m_{n-k,i}-n}/z_{2}\Big)^{-\hbar/\epsilon}.
$$
with $m_{k,i}=2 i -1$ and $m_{n-k,i}=n- 2k +2 i-1$, therefore
$$
\lim\limits_{\hbar\to \infty}\, \tilde{L}(z_{1},y_{k,i} \hbar^{m_{k,i}}) = 
\left\{\begin{array}{ll}
e^{\frac{1}{\epsilon} \frac{z_{1}}{y_{k,1}}}, & i=1,\\
1& i\neq 1
\end{array}\right.
$$
and
$$
\lim\limits_{\hbar\to \infty}\, 
\tilde{L}(z_{2} \hbar^{n},y_{n-k,i} \hbar^{m_{n-k,i}}) = 
\left\{\begin{array}{ll}
e^{\frac{1}{\epsilon} \frac{y_{n-k,k}}{z_{2}}}, & i=k,\\
1& i\neq k
\end{array}\right.
$$
Finally
$$
\Delta(y_{i,a} \hbar^{m_{i,a}},y_{i,b} \hbar^{m_{i,b}})=(1-y_{i,a}/y_{i,b}/(\hbar^{m_{i,b}-m_{i,a}}))^{2\hbar/\epsilon}
$$
and since $m_{i,b}-m_{i,a}\geq 2$ for $b>a$ we have
$$
\lim\limits_{\hbar\to \infty}\, \Delta(y_{i,a} \hbar^{m_{i,a}},y_{i,b} \hbar^{m_{i,b}})=1.
$$
In summary, the limit $\hbar\to \infty$ of a factor in (\ref{phiy}) is non-trivial only if it corresponds to an edge of the graph $\Gamma$ and:
\bean \label{pwlim}
\lim\limits_{\hbar\to \infty} \, \tilde{\Phi}(y,z)  = \prod_{e\in \textrm{edges}(\Gamma)} \, e^{\frac{1}{\epsilon}\, \frac{y_{h(e)}}{y_{t(e)}}} = e^{\frac{S({y},z)}{\epsilon}}
\eean
Finally, the point-wise limit (\ref{pwlim}) on the compact set $\gamma_{k,n}$ is uniform, the limit commutes with the integration and  from (\ref{verchang}) we obtain
\bea
\lim\limits_{\hbar \to \infty} \Ver(z/\hbar^{n})
&=&\lim\limits_{\hbar \to \infty} 
\dfrac{1}{(2 \pi \sqrt{-1})^{k(n-k)}} \oint\limits_{\gamma_{k,n}}\, \tilde{\Phi}(y,z) \, \bigwedge\limits_{i,j} \dfrac{dy_{i,j}}{y_{i,j}}
\\
& =&
\dfrac{1}{(2 \pi \sqrt{-1})^{k(n-k)}} \oint\limits_{\gamma_{k,n}}\, e^{\frac{S({y},z)}{\epsilon}} \, \bigwedge\limits_{i,j} \dfrac{dy_{i,j}}{y_{i,j}}\,,
\eea
which completes the proof. 
\end{proof}

\section{ Dwork congruences \label{DworkCong}}

Let $p$ be a prime number. Consider the power series
\bean 
\label{melser}
\mathsf{F}(z)= \sum\limits_{d=0}^{\infty}\, \left[S({x},{z})^d\right]_{0} \in \mathbb{Z}[[z]]\,.
\eean
Note  that the coefficients of this series  differ
from the coefficients  of the $A$-series  in \eqref{expinteg}  by a factor  
$d! \epsilon^d$,   i.e., these two series are related by the Borel integral transform.

\vsk.2>

Let us consider a system of  polynomial truncations of the function $\mathsf{F}(z)$\,:
$$
 \mathsf{F}_{s}(z)= \sum\limits_{d=0}^{p^s-1}\, 
 \left[S({x},{z})^d\right]_{0} \in \mathbb{Z}[z], \qquad s=0,1,2,\dots\,.
$$
where as before we assume that $\left[S({x},{z})^0\right]_{0}=1$. 
\begin{thm}
\label{thm con}
We have the following  congruences:
$$
\dfrac{\mathsf{F}(z)}{\mathsf{F}(z^{p})} \equiv \dfrac{\mathsf{F}_{s}(z)}{\mathsf{F}_{s-1}(z^{p})}  \mod p^s\,,
\qquad s = 1, 2, \dots\,.
$$
In particular, the polynomials $\mathsf{F}_{s}(z)$ satisfy the Dwork type congruences:
$$
\dfrac{\mathsf{F}_{s+1}(z)}{\mathsf{F}_{s}(z^{p})} \equiv
\dfrac{\mathsf{F}_{s}(z)}{\mathsf{F}_{s-1}(z^{p})} \mod p^s, \qquad s=1,2,\dots\,.
$$
\end{thm}

\begin{proof}
The proof of these  congruences is based on the properties of the Newton polytope
\bea
\Delta(k,n)=\textsf{N}(S({x},{z})) \subset \mathbb{R}^{k(n-k)}
\eea
 of the Laurent polynomial~(\ref{polsuper}) in the
variables $x=(x_{i,j})$. Let $f_{i,j}$ 
with $i=1,\dots,k, \ \ j=1,\dots, n-k$ denote the standard basis in $\mathbb{R}^{k(n-k)}$. 
The vectors
 $f_{i,j}$ correspond to the boxes of the $k\times (n-k)$-diagram 
 in Fig. \ref{boxgraph}. From (\ref{polsuper}), we see that $\Delta(k,n)$ is the convex hull of the vectors:
$$
f_{1,1}, \ \ f_{i,j+1}-f_{i-1,j+1}, \ \ i=2,\dots, k, \ \ j=0,\dots, n-k-1,
$$
$$
-f_{k,n-k}, \ \  f_{i,j+1}-f_{i,j}, \ i=1,\dots, k,\ \  j=1,\dots, n-k-1.
$$
This polytope has been considered in many publications, in particular it is known 
to be {\it reflexive} see Theorem 3.1.3 in \cite{BCKS98}. We recall
that the origin $(0,\dots,0)$ is the only integral point in a reflexive polytope, 
see for instance \cite{Nill} for an overview.

\medskip
Now the proof of Theorem \ref{thm con} follows from Theorem 1.1 in \cite{MeVl16}, after simple 
modifications. Let $S({x},1)$ denote the superpotential (\ref{polsuper}) 
with $z_{1}=z_{2}=1$. Clearly, this Laurent polynomial has the  same Newton polytope
 $\mathsf{N}(S({z},1))=\mathsf{N}(S({x},{z}))=\Delta(k,n)$.

 Consider
 $$
\mathsf{M}(\xi)=\sum\limits_{d=0}^{\infty}\, [S({x},1)^{d}]_{0}\, \xi^{d}, 
\qquad
\mathsf{M}_{s}(\xi)=\sum\limits_{d=0}^{p^s-1}\, [S({x},1)^{d}]_{0}\, \xi^{d}.
$$
Since $(0,\dots,0)$ is the only interior point of \ $\mathsf{N}(S({z},1))$ we may apply 
Theorem 1.1 in \cite{MeVl16} and conclude that these functions satisfy the congruences:
\bean 
\label{intiden}
\dfrac{\mathsf{M}(\xi)}{\mathsf{M}(\xi^{p})} \equiv \dfrac{\mathsf{M}_{s}(\xi)}{\mathsf{M}_{s-1}(\xi^{p})} \,,
\qquad
 \dfrac{\mathsf{M}_{s+1}(\xi)}{\mathsf{M}_{s}(\xi^{p})} \equiv
\dfrac{\mathsf{M}_{s}(\xi)}{\mathsf{M}_{s-1}(\xi^{p})} \mod p^s\,, \qquad s = 1, 2,\dots\,.
\eean
From the structure of the superpotential $S({x},z)$ it is clear that
$$
[S({x},z)^d]_{0}= [S({x},1)^d]_{0} z^{\frac{d}{n}}\,.
$$
In particular, this coefficient equals zero unless $n$ divides $d$. From this we find that
$$
\mathsf{F}(z)= \sum\limits_{d=0}^{\infty}\, \left[S({x},1)^d\right]_{0} z^{\frac{d}{n}} = \mathsf{M}(z^{\frac{1}{n}})
$$
and similarly $\mathsf{F}_s(z)=\mathsf{M}_s(z^{\frac{1}{n}})$. Now theorem \ref{thm con}
follows from (\ref{intiden}) after the substitution $\xi\to z^{\frac{1}{n}}$. 
\end{proof}

For further discussion of Dwork type congruences for vertex functions and solutions of 
KZ equations 
we refer to \cite{SmV23,Var22b,SV19,VZ21}.

\begin{rem}
Theorem \ref{thm con} implies an infinite factorization:
$$
\mathsf{F}(z)= \prod\limits_{i=0}^{\infty}\, \dfrac{\mathsf{F}_{s}(z^{p^i})}{\mathsf{F}_{s-1}(z^{p^{i+1}})} \mod p^s,
$$
cf. Theorem 5.3 in \cite{SmV23}.
\end{rem}

\end{document}